\documentclass[letterpaper, 10 pt, conference]{ieeeconf}

\IEEEoverridecommandlockouts                              

\title{\LARGE \bf
Sum-of-Squares Data-driven Robustly Stabilizing and Contracting Controller Synthesis for Polynomial Nonlinear Systems
}

\author{Hamza~El-Kebir$^{1}$ and Melkior~Ornik$^{2}$
\thanks{$^{1}$H. El-Kebir is with the Beckman Institute for Advanced Science and Technology, University of Illinois Urbana-Champaign, Urbana,
IL 61801, USA
{\tt\small elkebir2@illinois.edu}. Corresponding author.}%
\thanks{$^{2}$M. Ornik is with the Department
of Aerospace Engineering and the Coordinated Science Laboratory, University of Illinois Urbana-Champaign, Urbana,
IL 61801, USA
{\tt\small mornik@illinois.edu}}%
\thanks{Research reported in this publication was supported by the National Aeronautics and Space Administration under award number 80NSSC21K1030, and the Air Force Office of Scientific Research under award number FA9550-23-1-0131.}
}

\let\labelindent\relax

\usepackage{amsmath}
\usepackage{mathrsfs}
\usepackage{graphicx}
\usepackage[inline]{enumitem}
\usepackage{xcolor}

\newcommand{\hl}[1]{#1}
\newcommand{\hil}[1]{#1}
\newcommand{\hili}[1]{#1}
\newcommand{\hilir}[1]{#1}

\usepackage{bm}
\renewcommand{\vec}[1]{\boldsymbol{\mathbf{#1}}}

\usepackage{amsthm}

\newtheorem{theorem}{Theorem}
\newtheorem{lemma}{Lemma}

\newtheorem{proposition}{Proposition}

\theoremstyle{definition}
\newtheorem{definition}{Definition}

\newtheorem{assumption}{Assumption}

\theoremstyle{remark}
\newtheorem{remark}{Remark}

\newtheorem{corollary}{Corollary}

\usepackage[T1]{fontenc}
\usepackage{calligra}

\usepackage{stix}

\usepackage{dsfont}
\usepackage{hyperref}
\usepackage[inline]{enumitem}
\usepackage{graphicx}
\usepackage{caption}
\usepackage{subcaption }
\usepackage{tabularx}
\usepackage{booktabs}

\newcommand{\dd}{\mathrm{d}}

\newcommand{\transp}{^\mathsf{T}}
\newcommand{\inv}{^{-1}}
\newcommand{\pinv}{^{\dagger}}
\newcommand{\itransp}{^{-\mathsf{T}}}

\newcommand{\oslip}{\mathrm{osLip}}

\usepackage{scalerel, stackengine}

\stackMath
\newcommand{\dhat}[1]{\ThisStyle{\setbox0=\hbox{$\SavedStyle#1$}%
  \stackengine{0pt}{\SavedStyle#1}{\SavedStyle\hspace{.2\ht0}%
  \hat{\vphantom{#1}}\kern\dimexpr2.2\LMpt+.7pt\relax\hat{\vphantom{#1}}}{O}{c}{F}{T}{L}}%
}

\newcommand{\dcheck}[1]{\ThisStyle{\setbox0=\hbox{$\SavedStyle#1$}%
  \stackengine{0pt}{\SavedStyle#1}{\SavedStyle\hspace{.2\ht0}%
  \check{\vphantom{#1}}\kern\dimexpr2.2\LMpt+.7pt\relax\check{\vphantom{#1}}}{O}{c}{F}{T}{L}}%
}

\newcommand{\hatcheck}[1]{\ThisStyle{\setbox0=\hbox{$\SavedStyle#1$}%
  \stackengine{0pt}{\SavedStyle#1}{\SavedStyle\hspace{.2\ht0}%
  \hat{\vphantom{#1}}\kern\dimexpr2.2\LMpt+.7pt\relax\check{\vphantom{#1}}}{O}{c}{F}{T}{L}}%
}

\newcommand{\checkhat}[1]{\ThisStyle{\setbox0=\hbox{$\SavedStyle#1$}%
  \stackengine{0pt}{\SavedStyle#1}{\SavedStyle\hspace{.2\ht0}%
  \check{\vphantom{#1}}\kern\dimexpr2.2\LMpt+.7pt\relax\hat{\vphantom{#1}}}{O}{c}{F}{T}{L}}%
}

\newcommand{\DD}{\vec{\mathrm{D}}}
\newcommand{\hslashslash}{%
  \raisebox{.9ex}{%
    \scalebox{.7}{%
      \rotatebox[origin=c]{0}{$-$}%
    }%
  }%
}
\newcommand{\deltaslash}{%
  {%
   \vphantom{d}%
   \ooalign{\kern.05em\smash{\hslashslash}\hidewidth\cr$\delta$\cr}%
   \kern.05em
  }%
}

\renewcommand{\vec}[1]{\boldsymbol{\mathbf{#1}}}

\usepackage{stmaryrd}

\begin{document}

\maketitle
\thispagestyle{empty}
\pagestyle{empty}

\begin{abstract}
This work presents a computationally efficient approach to data-driven robust contracting controller synthesis for polynomial control-affine systems based on a sum-of-squares program. In particular, we consider the case in which a system alternates between periods of high-quality sensor data and low-quality sensor data. In the high-quality sensor data regime, we focus on robust system identification based on the data informativity framework. In low-quality sensor data regimes we employ a robustly contracting controller that is synthesized online by solving a sum-of-squares program based on data acquired in the high-quality regime, so as to limit state deviation until high-quality data is available. This approach is motivated by real-life control applications in which systems experience periodic data blackouts or occlusion, such as autonomous vehicles undergoing loss of GPS signal or solar glare in machine vision systems. We apply our approach to a planar unmanned aerial vehicle model subject to an unknown wind field, demonstrating its uses for verifiably tight control on trajectory deviation.
\end{abstract}


%
\IEEEpeerreviewmaketitle

\section{Introduction}

%

Data-driven identification and control has long been of interest as exemplified by the need for adaptive and robust control techniques. The quality of identified system models is strongly dependent on the available output data and its properties, making it essential to account for data integrity in estimation \cite{Eising2023}. This is even more so the case when synthesizing controllers based on output data, as is the case in the data-informativity framework \cite{VanWaarde2020}, where controllers can in some cases be synthesized \emph{directly from data}, bypassing classical synthesis structures involving sequential system identification and controller synthesis. While the development of stabilizing controllers based directly on data is useful, these approaches often leave much desired in the way of robustness and other operationally essential system theoretic properties.

In this work, we are particularly interested in ensuring local \emph{contractivity}, wherein trajectories starting at different initial states exponentially contract to one central trajectory \cite{Lohmiller1998, Lohmiller2000}. This notion is distinct from stability, where a stable system does not necessarily exhibit contraction other than in the linear case, and contractive systems are not necessarily stable \cite{Lohmiller1998, Aminzare2014a}. While past work has considered the effect of system inputs \cite{Sontag2010a} and small disturbances \cite{Margaliot2016}, little work exists in the way of synthesizing robustly contracting controllers. Namely, \cite{Tsukamoto2021} develops a learning-based approach to robust tube-based planning with contraction guarantees, an approach that necessitates computationally demanding intermediate machine learning prior to planning. Focusing on a control barrier function approach, \cite{Lopez2021} presents a data-driven method to ensure contraction; in doing so, a rigid nominal controller structure is demanded to allow for control contraction metrics to be computed offline in advance, severely restricting the applicable class of unknown systems.

To motivate our work, we consider systems that experience output model switching. This form of switching is often seen in practical applications wherein systems experience a loss of communication, line of sight, or otherwise lose access to (external) sensor information \cite{El-Kebir2024b}. For instance, in the case of unmanned aerial vehicles (UAVs), GPS denial can often be experienced when traversing hostile areas, building structures, or bodies of water \cite{Balamurugan2016}. In such cases, internal sensors such as inertial measurement units (IMUs) must be relied on, which often accumulate error due to sensor drift, therefore rendering sensor fusion approaches involving IMUs and GPS the leading approach to UAV state estimation. 
On account of degraded sensor estimation, set-propagation-based guaranteed observers \cite{Althoff2021a} and trajectory deviation integral inequalities \cite{El-Kebir2023} often produce exponentially expanding reachable sets, making it prohibitive to ensure safety across large periods of time without altering the controller structure to endow the ensemble of possible system realizations with contraction. Curbing this exponentially expansivity through control action is precisely the focus of the present work.

\begin{figure*}[h]
	\centering
	\includegraphics[width=\linewidth]{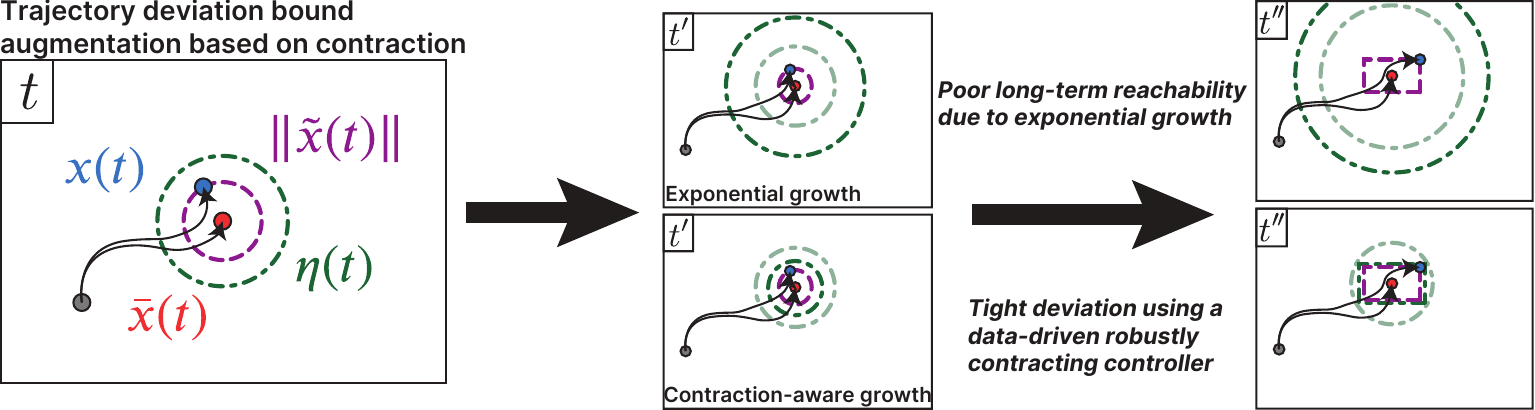}
	\caption{\hil{Overview of the data-driven robustly contracting controller approach (bottom), compared to a contraction-unaware approach based on expanding integral inequalities (top).}}
	\label{fig:overview}
\end{figure*}

In this work, we present a novel data-driven sum-of-squares approach to ensure robust stability and contraction over a compact set $K$ for a class of polynomial control-affine polynomial systems. Our approach is based on recent results due to Eising and Cortes \cite{Eising2023} and Davydov \emph{et al.} \cite{Davydov2022} on data-driven contraction and weighted Euclidean contraction conditions, respectively, as well as a well-known result due to Rantzer \cite{Rantzer2001, Zakeri2014} that forms a dual to Lyapunov's stability theorem. By ensuring stability of one nominal or learned system realization on $K$ and data-driven robust contraction on all potentially identifiable system realizations on $K$ based on the data noise characteristics, we synthesize a locally robustly stabilizing and contracting controller in spite of noise characteristics. Our approach is illustrated in Fig.~\ref{fig:overview}. We demonstrate our approach on a polynomial nonlinear system featuring a UAV's planar dynamics, showing how our approach can formulate explicit data-driven conditions on contraction rates using an efficient sum-of-squares approach.

\section{Preliminaries}

Given a \hl{non-singular} matrix $R \in \mathbb{R}^{n \times n}$, let $\Vert x \Vert_{p, R} := \Vert R x \Vert_p$ denote the weighted $p$-norm. For a matrix $B \in \mathbb{R}^{n \times m}$, let $B\pinv$ denote its Moore--Penrose pseudo-inverse. Let $\mathcal{K}(\mathbb{R}^n)$ [$\mathcal{KC}(\mathbb{R}^n)$] denote the family of compact [compact convex] subsets of $\mathbb{R}^n$. Let $\mathcal{B}_r^n$ be \hl{the zero-centered} closed ball with radius $r > 0$ in $\mathbb{R}^n$.
Let $\mathbb{S}^n$ denote the space of symmetric real $n \times n$ matrices. If $A \in \mathbb{S}^{n + m}$, when $n + m$ clearly follow from context, we can partition $A$ as
$
	A = \begin{bmatrix}
		A_{11} & A_{12} \\
		A_{21} & A_{22}
	\end{bmatrix},
$
where $A_{11} \in \mathbb{S}^n$ and $A_{22} \in \mathbb{S}^m$. We denote the Schur complement of $A$ with respect to $A_{22}$ by $A | A_{22} := A_{11} - A_{12} A_{22}\pinv A_{21}$. Such a partitioned matrix $A$ gives rise to the following set defined by a quadratic matrix inequality (QMI):
\begin{equation*}
	\mathcal{Z}(A) := \left\{ Z \in \mathbb{R}^{m \times n} : \begin{bmatrix} I_n \\ Z \end{bmatrix}\transp A \begin{bmatrix} I_n \\ Z \end{bmatrix} \geq 0 \right\}.
\end{equation*}

We define the \emph{logarithmic norm} given a norm $\Vert \cdot \Vert_p$ on $\mathbb{R}^n$ acting on a linear operator $A \in \mathbb{R}^{n \times n}$ \cite{Soderlind2006} as:
\begin{equation}
	\mu_{p, R} (A) := \lim_{h \searrow 0} \frac{\Vert R( I_n + h A) \hl{R\pinv} \Vert_p - 1}{h}.
\end{equation}
\hl{For any symmetric matrix $P$, let $P > 0$ denote positive definiteness.} Recall that the logarithmic weighted Euclidean norm for a weighting matrix $P^{1/2}$ where $P = P\transp > 0$, is characterized \cite[Cor.~7]{Davydov2022} as:
\begin{equation}\label{eq:weighted log norm}
	\mu_{2, P^{1/2}} (A) = \lambda_{\max} \left( \frac{P A P\inv + \hl{A\transp}}{2} \right).
\end{equation}

For a system $\dot{x} = f(x)$, \hl{from \eqref{eq:weighted log norm} and the Gr\"onwall--Bellman lemma, it follows that $\Vert x(t) - x'(t) \Vert_{p, P^{1/2}} \leq \Vert x(0) - x'(0) \Vert_{p, P^{1/2}} \exp(\gamma t)$ for $\gamma := 
\sup_{x \in \mathbb{R}^n} \mu_{p, P^{1/2}}(\DD f(x))$ \cite[Thm.~29]{Davydov2022}}. Here, the system is \emph{contracting} if $\gamma < 0$ and \emph{expanding} if $\gamma > 0$. We denote the total derivative of a differentiable function $f : \mathbb{R}^n \to \mathbb{R}^n$ at $x \in \mathbb{R}^n$ by $\DD f(x) \in \mathbb{R}^{n \times n}$ such that $\lim_{x' \to x} \Vert f(x') - f(x) - \DD f(x)(x-x') \Vert / \Vert x' - x \Vert = 0$.

We briefly introduce the data informativity framework \cite{Eising2023}, which will enable us to produce robust contraction certificates later on in this work.

\subsection{Data-driven Identification}

We consider a system of the form:
\begin{equation}\label{eq:nom sys}
	\dot{x}(t) = f(x(t)) + g(x(t)) u, \quad x(0) = x_0 \in \mathbb{R}^n,
\end{equation}
where $f : \mathbb{R}^n \to \mathbb{R}^n$ is an unknown Lipschitz continuous function, $g : \mathbb{R}^n \to \mathbb{R}^{n \times m}$ is a \textit{known} continuous map. 

\hil{We are interested in identifying function $f$ in the case of noisy observations $y = f(x) + w$, where a set of noise inputs $W \ni w$ is explicitly considered in the construction of a least-squares estimator for $f$. This latter incorporation of a deterministic set of noise inputs is key in this work, since it will enable us to incorporate the effect of measurement uncertainty on our controller synthesis scheme, allowing for \emph{immediate data-driven construction of noise-robust controllers}. \hili{In this formulation, we effectively assume noisy knowledge of $\dot{x}$, $g(x)$ and $u$.}}

\hil{At the end of this section, we address the problem of guaranteeing contractivity for an unknown function $f$ in \eqref{eq:nom sys} that is identified on the basis of noisy measurements $y_i = f(x_i) + w_i$, for $w_i$ in some known noise set $W$.}


Let us assume that $b : \mathbb{R}^n \to \mathbb{R}^k$ is a given vector-valued basis function comprised of a family of basis functions $\phi_i : \mathbb{R}^n \to \mathbb{R}$ for $i=1,\ldots,k$:
\begin{equation}\label{eq:basis functions}
	b(x) := \begin{bmatrix}
		\phi_1(x) & \cdots & \phi_k(x)
	\end{bmatrix}\transp,
\end{equation}
such that $f(x) = \hat{\phi}(x) := \hat{\theta}\transp b(x)$ for all $x \in \mathbb{R}^n$, for some $\hat{\theta} \in \mathbb{R}^{k \times n}$.
Based on these basis functions, we define $\phi_\theta : \mathbb{R}^n \to \mathbb{R}^n$ parameterized by $\theta \in \mathbb{R}^{k \times n}$ as $\phi_\theta (x) = \theta\transp b(x)$. It follows from our assumptions that $\hat{\phi} = \phi_{\hat{\theta}}$.

Let us assume that we collect noisy measurements $y_i = \hat{\phi}(x_i) + w_i$ for $i=1,\ldots,T$ at points $\{x_i\}_{i=1}^{T}$. Note that $y_i$ here \emph{does not} denote the output of a system; in fact, we shall consider $y_i = f(x_i) + w_i = \hat{\phi}(x_i) + w_i$ in this work. The set of data can be expressed in matrix form as follows:
\begin{equation}
\begin{split}
	&Y := \begin{bmatrix}
		y_1 & \cdots & y_T
	\end{bmatrix}, \quad
	W := \begin{bmatrix}
		w_1 & \cdots & w_T
	\end{bmatrix}, \\
	&\Phi := \begin{bmatrix}
		b(x_1) & \cdots & b(x_T)
	\end{bmatrix} = \begin{bmatrix}
		\phi_1 (x_1) & \cdots & \phi_1 (x_T) \\
		\vdots & \ddots & \vdots \\
		\phi_k (x_1) & \cdots & \phi_k (x_T)
	\end{bmatrix},
\end{split}
\end{equation}
from which it follows that $Y = \hat{\theta}\transp \Phi + W$, where $Y$ and $\Phi$ are known, and $\hat{\theta}$ and $W$ are unknown. We make the following assumption about the noise model:

\begin{assumption}[Noise model {\cite[Assumption~1]{Eising2023}}]\label{assump:noise model}
	We assume that the noise samples $W$ satisfy $W\transp \in \mathcal{Z}(\Pi)$, where $\Pi \in \mathbb{S}^{n + T}$ is such that $\Pi_{22} < 0$ and $\Pi | \Pi_{22} \geq 0$, and $\Pi$ is fixed and known. \hl{We also assume that $T \geq n$, such that $\Phi$ is of full column rank.}
\end{assumption}

Under Assumption~\ref{assump:noise model}, the set of parameters consistent with measurements $Y$ can be expressed as
\begin{equation}
	\Theta = \{ \theta \in \mathbb{R}^{k \times n} : \exists W\transp \in \mathcal{Z}(\Pi), Y = \theta\transp \Phi + W \},
\end{equation}
from which we can write
\begin{equation*}
	Y = \theta\transp \Phi + W \Leftrightarrow
	\begin{bmatrix}
		I_n & W	
	\end{bmatrix}
	=
	\begin{bmatrix}
		I_n & \theta\transp	
	\end{bmatrix}
	\begin{bmatrix}
		I_n & Y \\
		0 & -\Phi	
	\end{bmatrix},
\end{equation*}
from which it follows that $\Theta = \mathcal{Z}(N)$, where $N$ is defined as
\begin{equation}\label{eq:N definition}
	N :=
	\begin{bmatrix}
		I_n & Y \\
		0 & -\Phi	
	\end{bmatrix}
	\Pi
	\begin{bmatrix}
		I_n & Y \\
		0 & -\Phi	
	\end{bmatrix}\transp.
\end{equation}

Hence, we can obtain a least-squares estimate of $\hat{\phi}(x)$ as $\phi^{\mathrm{lse}}(x) = (\theta^{\mathrm{lse}})\transp b(x) = -N_{12} N_{22}\pinv b(x)$ \cite[Remark~2.2]{Eising2023}. \hil{Note that this differs from classical least squares, where the effect of a noise set $\mathcal{Z}(\Pi)$ is not taken into account, leading to biased and incorrect estimates of the system dynamics. The result presented here explicitly accounts for the noise set, thereby providing a consistent framework that allows us to quantify the effect on noise on the identified system's properties, including contraction as shown next.}

%

We now \hl{state} a preliminary result on data-driven strict contractivity from \cite{Eising2023} in terms of the \emph{local one-sided Lipschitz constant}:

\begin{definition}[Local One-sided Lipschitz Constant]\label{def:loc osLip}
	For a function $f : \mathbb{R}^n \to \mathbb{R}^n$, a matrix $P > 0$, and a compact set $K \in \mathcal{K}(R^n)$, $f$ is \emph{one-sided Lipschitz with respect to} $\Vert \cdot \Vert_{2, P^{1/2}}$ if there exists $L$ such that, for all $x, x' \in K$,
	\begin{equation}\label{eq:local osLip}
		(x - x')\transp P (f(x) - f(x')) \leq L \Vert x - x' \Vert_{2, P^{1/2}}^2,
	\end{equation}
	where the smallest such $L$ is denoted by $\oslip_K (f)$.
\end{definition}

\begin{remark}
	Note that $\oslip_K (f)$ is not necessarily positive unlike the classical Lipschitz constant, e.g., for $P = I$ and $f(x) = -x$, we have $\oslip_K (f) = -1$ for any $K \in \mathcal{K}(\mathbb{R}^n)$. Indeed, $f$ is \emph{strictly contracting} at rate $|\gamma|$ if $\oslip_K (f) < \gamma$, where $\gamma < 0$; $f$ is said to be \emph{strictly expanding} if $\gamma > 0$.
\end{remark}

Based on Definition~\ref{def:loc osLip}, we can formulate a local robust contractivity result based on the data-driven least-square estimate of the dynamics. \hil{It is important to note that this result will play a key role in the remainder of this paper, since it establishes a criterion that guarantees mutual contraction for systems that are identified in a data-driven fashion under noise.}

\begin{proposition}[Local Data-driven Contractivity]\label{prop:local data-driven contractivity}
	Given data $(Y, \Phi)$, let $\Theta = \mathcal{Z}(N)$. Assume $\Phi$ has full \hl{column} rank. Given a set $K \in \mathcal{K}(\mathbb{R}^n)$ and $P > 0$, we have
	$
		(x - x')\transp P \theta\transp (b(x) - b(x')) \leq \gamma \Vert x - x' \Vert_{2, P^{1/2}}^2
	$
	for all $\theta \in \Theta$ and $x, x' \in K$ if
	\begin{equation}
		\oslip_K (f) < \gamma - L_K \sqrt{\lambda_{\max} (N|N_{22})} \frac{\lambda_{\max} (P)}{\lambda_{\min} (P)},
	\end{equation}
	where $L_K$ is such that $\Vert b(x) - b(x') \Vert_{2, (-N_{22})^{-1/2}} \leq L_K \Vert x - x' \Vert_2$ for all $x, x' \in K$.
\end{proposition}

\begin{proof}
	The proof follows by \cite[Cor.~4.2]{Eising2023}, which we extended to the case of local strict contractivity using local one-sided Lipschitz constants from Definition~\ref{def:loc osLip}.
\end{proof}

\begin{remark}
	Proposition~\ref{prop:local data-driven contractivity} implies that systems with dynamics $\dot{x} = \theta\transp b(x)$ for any $\theta \in \Theta$ robustly contracts with rate $|\gamma|$ for $\gamma$ for all $x \in K$, based on a condition on $\phi^{\mathrm{lse}}(x) = (\theta^{\mathrm{lse}})\transp b(x)$. \hili{The assumption that $\Phi$ is full rank is not restrictive in practice, as it demands persistent excitation in the observed trajectory $\{x_1, \ldots, x_T\}$, since constant $x$ will cause this assumption to break down. This is a very common assumption in system identification problems and adaptive control problems.}
\end{remark}

We now proceed by briefly introducing the notion of sum-of-squares programming, which will be the framework in which we obtain our controller synthesis approach.

\subsection{Sum-of-Squares Programming} We briefly define what it means for a polynomial to be a sum of squares, as well as what we mean by sum-of-squares (SOS) programs, a class of convex optimization problems solvable using semidefinite programming techniques:

\begin{definition}[Sum-of-Squares {\cite[Def.~3]{Zakeri2014}}]
	A polynomial $p(x)$ is said to be a \emph{sum of squares} (SOS) if there exist polynomials $q_i (x)$ such that $p(x) = \sum_i q_i^2(x)$.
	Let $\mathcal{S}_n$ denote the set of all SOS polynomials, and $\mathcal{P}_n$ the set of all polynomials from $\mathbb{R}^n$ to $\mathbb{R}^n$. 
\end{definition}

\begin{remark}
	Note that $p(x)$ being SOS implies $p(x) \geq 0$ for all $x$, but the inverse does not hold in general (see T.~S. Motzkin's famous inequality \cite{Motzkin1967}).
\end{remark}

\begin{definition}[Sum-of-Squares Program {\cite[Def.~4]{Zakeri2014}}]
	A sum-of-squares program is a convex optimization program of the following form:
	
	Given a weighting vector $w$ and polynomials $\{a_{i,j}\}_{i,j} \subseteq \mathcal{P}_n$, find $l$ polynomials $\{p_{i}\}_{i=1}^{\hat{l}} \subseteq \mathcal{P}_n$ and $\{p_{i}\}_{i=\hat{l}+1}^{l} \subseteq \mathcal{S}_n$ with a vector of coefficients $c$ that minimizes $w\transp c = \sum_{i=1}^l w_i c_i$
	such that the equality constraints $a_{0,j} + \sum_{i=1}^l \hl{p_i} (x) a_{i,j} (x) = 0$,
	for $j = 1,\ldots, \hat{k}$ and the SOS constraints $a_{0,j} + \sum_{i=1}^l \hl{p_i} (x) a_{i,j} (x) \in \mathcal{S}_n$
	for $j = \hat{k}+1, \ldots, k$ are satisfied.
\end{definition}

In what follows, due to space constraints and for ease of exposition, we will consider \emph{polynomial} control-affine systems of the following form:
\begin{equation*}
	\dot{x}(t) = f(x(t)) + B u(t), \quad x(0) = x_0 \in \mathbb{R}^n, u \in U \subseteq \mathbb{R}^m,
\end{equation*}
where $f(x)$ lies in the span of chosen polynomial basis functions $b(x)$, i.e., there exists $\hat{\theta} \in \Theta$ with $f(x) = \hat{\theta}\transp b(x)$. We can for instance take monomials up to a known order to construct $b(x)$. Most of our results can be extended to a scenario where $B$ is not constant; such a case will be the subject of a future publication. We now consider the problem of data-driven \emph{robustly contracting} controller synthesis.

\section{Data-driven Contracting Controller Synthesis}

We first show a data-driven robust contraction criterion based on a given fixed full-state feedback gain $G \in \mathbb{R}^{m \times n}$. Robustness in this case arises from the fact we account for errors in estimation due to noise $W$.

Before stating such a result, we emphasize that we do not have access to the true state in the low-fidelity data regime. To account for this discrepancy, we study the difference between a feedback control signal generated by the true state $u_* = G_* x_*$ and by an erroneous state estimate $u = G_* x$. This state estimation error results in an additional robustness criterion regarding discrepant feedback signals, complementing the results above which account for estimation errors in the dynamics. \hl{Data $(Y, \Phi)$ is only gathered in the high-fidelity regime; the control scheme proposed in this section will then be active during the low-fidelity data regime.}

We also introduce an additional stability constraint, restricting the class of feedback gains to those produced by an as of yet unknown map $\mathrm{fdbk}_K$, which takes the estimate of the free dynamics $\phi^{\mathrm{lse}}$ and produces a set of admissible stabilizing feedback gains on the set $K$. Later in this work, we shall explicitly show how to obtain constraints that result in $G$ being a stabilizing controller.

\begin{theorem}\label{thm:contracting controller given est state}
	Given data $(Y, \Phi)$, let $\Theta = \mathcal{Z}(N)$. Let $P > 0$ and $K \in \mathcal{KC}(\mathbb{R}^n)$. \hl{Assume that Assumption~\ref{assump:noise model} holds} and assume that in \eqref{eq:nom sys}, $g(x) = B$ for some known $B \in \mathbb{R}^{n \times m}$ for all $x \in K$. 
Let
	\begin{equation}\label{eq:L_K estimate}
		L_K := \sqrt{-1/\lambda_{\max} (N_{22})} \max_{x \in K} \Vert \DD b(x) \Vert_2.
	\end{equation}
	
	Given a desired contraction rate $\gamma < 0$ and
	\begin{equation}\label{eq:gamma star}
		G_* := \arg\min_{G \in \mathrm{fdbk}_K (\phi^{\mathrm{lse}})} \lambda_{\max} \left( \frac{P B G P\inv + G\transp B\transp}{2} \right),
	\end{equation}
	if $\mu_* := \lambda_{\max} \left( \frac{P B G_* P\inv + G_*\transp B\transp}{2} \right)$ is such that 
	\begin{equation*}
		\mu_* < \gamma - L_K \sqrt{\lambda_{\max} (N|N_{22})} \frac{\lambda_{\max} (P)}{\lambda_{\min} (P)} - \ell - \oslip_K (\phi^{\mathrm{lse}}),
	\end{equation*}
	where $\ell := \sqrt{\lambda_{\max}(P)} |\lambda_{\max} (B G_*)|$,
	then the system $\dot{x} = f(x) + BG_* \bar{x}$ is locally contracting on $K$ at rate $|\gamma|$ for any \hl{potentially erroneous state estimate} $\bar{x} \in K$.
\end{theorem}

\begin{proof}
The proof of Theorem \ref{thm:contracting controller given est state} is chiefly based on Proposition~\ref{prop:local data-driven contractivity} and subadditivity of the one-sided Lipschitz constant \cite[Prop.~27]{Davydov2022}. 

To obtain $L_K$, we find the following estimate on $\Vert x \Vert_{2, P^{1/2}}$:
	\begin{equation*}
		\Vert x \Vert_{2, P^{1/2}}^2= (P^{1/2} x)\transp (P^{1/2} x) = x\transp P x \leq \lambda_{\max} (P) \Vert x \Vert_2^2.
	\end{equation*}
	The resulting expression for $L_K$, \eqref{eq:L_K estimate}, then follows by consideration of the mean-value theorem on a compact set $K$ based on the norm of the derivative of the basis function $b(x)$ (see \cite[\S 6.1]{Eising2023}). 
	Then we consider Proposition~27 and Theorem~38 of \cite{Davydov2022}, which introduces the additional one-sided Lipschitz bound $\ell = \sqrt{\lambda_{\max}(P)} |\lambda_{\max} (B G_*)|$ that influences the contraction rate:
	\begin{equation*}
		\Vert f(x) + B G_* x - (f(x) + BG_* \bar{x}) \Vert_{2, P^{1/2}} \leq \ell \Vert x - \bar{x} \Vert_2,
	\end{equation*}
	where we have taken $x, \bar{x} \in K$. In this instance, $\bar{x}$ can be considered the (potentially incorrect) full-state estimate obtain through a set-based observer scheme, where the guaranteed reachable set is taken as $K$ or a subset thereof. Note that the free dynamics are based on the unknown true state $x$, even when the feedback controller is based on the potentially false state estimate $\bar{x}$. Finally, $G_*$ is obtained by evaluating $\arg\min_{G \in \mathrm{fdbk}_K (\phi^{\mathrm{lse}})} \oslip_K (B G)$, which reduces to \eqref{eq:gamma star} by \eqref{eq:weighted log norm}.
\end{proof}

\begin{remark}
\hl{The use of a feedback law incorporating a potentially erroneous state estimate $\bar{x} \in K$ leads to $u = G_* \bar{x}$. Since both $x$ and $\bar{x}$ belong to $K$, it is possible to make the system $\dot{x} = f(x) + B G_* \bar{x}$ contracting, simply by considering $\dot{x} = f(x) + B G_* \bar{x} = f(x) + B G_* x + B v$, where $v := G_* (x - \bar{x}) \in G_* \mathcal{B}^n_{\mathrm{diam(K)/2}}$. Doing so is equivalent to considering an additive disturbance $v$ to the unknown nominal control input $u = G_* x$, in which case the conditions for contractivity follow as shown in Theorem~38 of \cite{Davydov2022}.}
	
	
	\hili{$G_*$ can be obtained based on a sum-of-squares (SOS) program, which we develop for a particular case in Theorem~\ref{thm:main synthesis sos}.}
\end{remark}

Based on these sufficient conditions for robust contraction under estimation-based loop closure, we proceed by developing an SOS program to synthesize robustly contracting and stabilizing feedback controllers based on the data acquired.

\section{Data-driven Contracting and Stabilizing Feedback Controller Synthesis}

Having obtained sufficient conditions for robust contraction as presented in Theorem~\ref{thm:contracting controller given est state}, our goal is now to construct a computationally tractable algorithm that is \emph{not dependent on nonlinear optimization}. We will show that it is possible construct an admissible feedback controller based on sum-of-squares (SOS) techniques. A robust proportional--integral controller synthesis technique for control-affine systems with bounded parametric uncertainties in given in \cite{Zakeri2014}, where the resulting sum-of-squares program can efficiently be solved; this work did not, however, consider contraction guarantees, and is not data-driven in nature. Classically, SOS problems are solved using equivalent semidefinite program (SDP) formulations. In recent years, techniques such as \emph{(scaled) diagonally dominant sum-of-squares} ((S)DSOS) have emerged \cite{Ahmadi2019}, enabling efficient online solution of SOS programs with tight control on execution time and error bounds, as well as domain specification for local optimization problems.

Before we present the algorithm, we raise a number of key results relating the conditions shown above to equivalent sum-of-squares constraints. First, we provide a sum-of-squares program to obtain the local one-sided Lipschitz bound of a polynomial function.

\begin{proposition}[Sum-of-squares Program for One-sided Lipschitz Bound]\label{prop:sos oslip}
	For a continuously differentiable polynomial function $f : \mathbb{R}^n \to \mathbb{R}^n$, matrix $P > 0$, and a set $K \in \mathcal{KC}(\mathbb{R}^n)$, $\oslip_K (f)$ is given by \hl{
	\begin{equation}
	\begin{split}
		\oslip_K (f) &= \min_\nu \nu \ \mathrm{s.t.} \\
		z_2\transp (2 \nu P - (P \DD f (z_1) &+ \DD\transp f(z_1) P\transp)) z_2 \in \mathcal{S},
	\end{split}
	\end{equation}
	where $z = [\begin{smallmatrix} z_1\transp & z_2\transp \end{smallmatrix}]\transp \in K \times \mathcal{B}_1^n$.}
\end{proposition}
\begin{proof}
	The proof follows by the Demidovi\v{c} condition for one-sided weighted Lipschitzness \cite{Demidovic1961}, which reads $\nu \geq \oslip (f)$ if and only if
	$
		P \DD f(x) + \DD\transp f(x) P\transp \leq 2 \nu P
	$ \cite[Table~1]{Davydov2022}.
	We then recast this condition into an equivalent sum-of-squares constraint, completing the proof. \hl{This sum-of-squares constraint needs only to hold on $K \times \mathcal{B}_1^n$ as opposed to $\mathbb{R}^{2n}$, forming a constrained SOS problem \cite{Ahmadi2019}.}
\end{proof}

We aim to construct a sum-of-squares program that produces a robust contracting and stabilizing fixed-gain feedback controller. Such a synthesis approach will be computationally tractable, while allowing for data-driven robust stabilization
. We first cite a key result from \cite{Rantzer2001} concerning a dual to Lyapunov's stability theorem, which enables the formulation of sum-of-squares constraints to certify stability:

\begin{lemma}[{\cite[Thm.~1]{Rantzer2001}}]\label{lm:Rantzer Lyap}
	Given $\dot{x}(t) = f(x(t))$, where $f \in \mathcal{C}^1 (\mathbb{R}^n, \mathbb{R}^n)$ and $f(0) = 0$, suppose there exists a nonnegative $\rho \in \mathcal{C}^1 (\mathbb{R}^n \setminus \{0\}, \mathbb{R})$ such that $\frac{\rho(x) f(x)}{\Vert x \Vert}$ is integrable on $\{ x \in \mathbb{R}^n : \Vert x \Vert \geq 1 \}$ and the divergence of $\rho(x) f(x)$,
	$
		[\DD \cdot (\hl{\rho} \cdot f)] (x) > 0, \ x \ \mathrm{almost \ everywhere}.
	$
	Then, for almost all initial states $x(0)$ the trajectory $x(t)$ exists for $t \in [0, \infty)$ and tends to zero as $t \to \infty$.
\end{lemma}

	\hil{One of the key innovations presented by our approach is an alternative approach to robustness. Classically, robust controller synthesis algorithms explicitly account for parameter uncertainty or exogenous disturbances, often demanding large-scale optimization problems that operate on coarse disturbance models. 
    Instead, in our approach, we only need the nominal system to be stabilized by the synthesized controller. Robustness stability is endowed on the class of all off-nominal systems under consideration through robust contraction towards a stabilized nominal system trajectory. 
    }
\hil{We provide a brief theorem that shows that a nominally stabilizing controller that is robustly contracting is also a robust stabilizing controller, following the discussion given in the previous remark.}

\begin{theorem}
\label{thm:robust contraction implies robust stability}
	\hil{Consider a system of the form $\dot{x}(t) = f_{\mathrm{nom}} (x(t), u(t))$, where $f$ is a continuous function in $x \in X \subseteq \mathbb{R}^n$ and $u \in U \subseteq \mathbb{R}^n$. Consider a set of possible dynamics $F := \{f_{\mathrm{nom}}\} + \mathcal{B}_1^{\Vert \cdot \Vert_\infty} \subseteq C^0(X \times U, \mathbb{R}^n)$. Let $u(t) = G x(t)$ be a controller that exponentially stabilizes $\dot{x}_{\mathrm{nom}}(t) = f_{\mathrm{nom}} (x_{\mathrm{nom}}(t), G x_{\mathrm{nom}}(t))$ at a rate $\alpha > 0$, such that
	\begin{equation*}
		\Vert x_{\mathrm{nom}}(t) \Vert \leq \Vert x_{\mathrm{nom}}(0) \Vert \exp(-\alpha t).
	\end{equation*}
	Let the control law $u(t) = G x(t)$ also be such that it is robustly contracting at a rate $\gamma > 0$ for all dynamics given by $\dot{x}(t) = f (x(t), u(t))$ for $f \in F$, i.e.,
	\begin{equation*}
		\oslip_X (f(\cdot, G \cdot)) \leq -\gamma < 0.
	\end{equation*}
	Then, for any $f \in F$, the system $\dot{x}(t) = f (x(t), u(t))$ with $u(t) = G x(t)$ is exponentially stable at a rate $\min\{\alpha, \gamma\}$, i.e.,
	\begin{equation*}
		\Vert x(t) \Vert \leq \left( \Vert x(0) - x_{\mathrm{nom}}(0) \Vert + \Vert x(0) \Vert \right) \exp(-\min\{\alpha, \gamma\} t).
	\end{equation*}
    }
\end{theorem}


\begin{proof}
	\hilir{Robust contraction at a rate $\gamma$ gives the following bound on trajectory convergence:
	\begin{equation*}
		\Vert x_{\mathrm{nom}}(t) - x(t) \Vert \leq \Vert x_{\mathrm{nom}}(0) - x(0) \Vert \exp(-\gamma t).
	\end{equation*}
	Since we are interested in exponential convergence of $x(t)$ to the origin, we find
	\begin{equation*}
	\begin{split}
		\Vert x(t) \Vert &= \Vert x(t) - x_{\mathrm{nom}}(t) + x_{\mathrm{nom}}(t) \Vert \\
		&\leq \Vert x(t) - x_{\mathrm{nom}}(t) \Vert + \Vert x_{\mathrm{nom}}(t) \Vert \\
		&\leq \Vert x(0) - x_{\mathrm{nom}}(0) \Vert \exp(-\gamma t) + \Vert x_{\mathrm{nom}}(0) \Vert \exp(-\alpha t) \\
		&\leq \left( \Vert x(0) - x_{\mathrm{nom}}(0) \Vert + \Vert x(0) \Vert \right) \exp(-\min\{\alpha, \gamma\} t),
	\end{split}
	\end{equation*}
	where we have applied the triangle inequality in conjunction with the exponential stability of $x_{\mathrm{nom}}(t)$ and the robust contraction of $x(t)$ to $x_{\mathrm{nom}}(t)$.}
\end{proof}

\begin{corollary}
	\hilir{We can show that the above result also holds in its converse form after a slight strengthening of the robust stability property. Indeed, one can directly show that robust \emph{exponential} stability implies robust contraction using a similar proof approach.}
\end{corollary}

We may now proceed by formulating a constrained sum-of-squares program to synthesize a robust contracting and stabilizing feedback controller for a class of affine-in-control polynomial systems: 

\begin{theorem}\label{thm:main synthesis sos}
\hl{Assume that the hypotheses of Theorem~\ref{thm:contracting controller given est state} hold.} Let $\mathrm{span} \ K = \mathbb{R}^n$ and let $B \in \mathbb{R}^{n \times m}$ of full column rank. Let the least squares estimate of $f$ based on data $(Y, \Phi)$ be given as $\phi^{\mathrm{lse}}$. Let $p$ be an arbitrary positive-definite polynomial over $K$ with values in $\mathbb{R}_+$, and $\alpha > 0$ a sufficiently large constant. Let $\oslip_K (\phi^{\mathrm{lse}})$ be obtained as in Proposition~\ref{prop:sos oslip}.

For a given desired contraction rate $\gamma < 0$, consider $\Gamma \in \mathbb{S}^n$ that solves the following constrained SOS program:
	\begin{align*}\tag{OPT}\label{eq:optimization problem}
		\min_{\mu, a, \Gamma} &\ \mu \ \mathrm{s.t.} \\
		\mu < \gamma/\beta - \varsigma/\beta, \quad \Gamma &= \Gamma\transp, \quad a > 0,
	\end{align*}
	\begin{equation}\tag{SOS-I}\label{eq:SOS-I}
	x\transp \left(\mu I - (P \Gamma P\inv + P\itransp \Gamma P\transp)/2\right) x \in \mathcal{S},
	\end{equation}
	\begin{equation}\tag{SOS-II}\label{eq:SOS-II}
		p(x) \DD \cdot (a \phi^{\mathrm{lse}}(x) + a \Gamma x)
		- \alpha \DD p(x) \cdot (a \phi^{\mathrm{lse}}(x) + a \Gamma x) \in \mathcal{S}
	\end{equation}
	for all $x \in K$, where $\beta := 1 + \sqrt{\lambda_{\max}(P)}$ and
	\begin{equation}\label{eq:sig}
		\varsigma := \left(-\frac{\lambda_{\max}(N | N_{22})}{\lambda_{\max} (N_{22})}\right)^{1/2} \max_{x \in K} \Vert \DD b(x) \Vert_2 + \oslip_K (\phi^{\mathrm{lse}}).
	\end{equation}
	If such a $\Gamma$ exists, then
	the resulting feedback law $u = G \hat{x}$ for $\hat{x} \in K$, where $G = B\pinv \Gamma / a$, is stabilizing and contracting at rate $|\gamma|$ on $K$ \hil{by Theorem~\ref{thm:robust contraction implies robust stability}}.
\end{theorem}

\begin{proof}
The proof follows from Theorem~\ref{thm:contracting controller given est state} and a modification of the proof of Theorem~3 of \cite{Zakeri2014}.

In brief, the first inequality constraint of \eqref{eq:optimization problem} ensures that the closed-loop system contracts at rate $|\gamma|$ on $K$, despite modeling uncertainty and state estimation error. Here, $\gamma$ is the desired contraction rate, and $\varsigma$ in \eqref{eq:sig} accounts for modeling uncertainty and the effect of the free dynamics; $\beta$ accounts for the state estimation error.

The first SOS constraint, \eqref{eq:SOS-I}, enforces the maximum eigenvalue constraint on $\Gamma$. The equality constraint of \eqref{eq:optimization problem} enforces symmetricity of $\Gamma$, such that, when combined with \eqref{eq:SOS-I},
	$
		\mu \geq \lambda_{\max} ((P \Gamma P\inv + \hl{\Gamma\transp})/2).
	$
	Indeed, the stipulation that $\mathrm{span} \ K = \mathbb{R}^n$ ensures that \eqref{eq:SOS-I} implies $v\transp \Gamma v \leq \mu v\transp v$ for all $v \in \mathbb{R}^n$, which in turn implies $\lambda_{\max} (\Gamma) \leq \mu$.
	
	The second SOS constraint, \eqref{eq:SOS-II}, ensures stability of the closed-loop system on $K$ on account of Lemma~\ref{lm:Rantzer Lyap}, \hl{where $\rho(x) = \frac{a}{p^\alpha (x)}$. The integrability constraint of Lemma~\ref{lm:Rantzer Lyap} implies that $a$ and $\alpha$ cannot be trivially taken to approach infinity.} If all conditions are satisfied, $G = B\pinv \Gamma/a$ produces a robustly stabilizing and contracting feedback gain, where robustness is with respect to both modeling errors and state estimation errors. It is trivial to show that $B\pinv$ is a left inverse if $B$ is of full column rank, as assumed in the hypothesis.
	
	Robustness to modeling errors follows by accounting for state estimation errors using the $\ell$ term in Theorem~\ref{thm:contracting controller given est state}. Despite the controller being synthesized based on $\phi^{\mathrm{lse}}$, robust contraction ensures stability on $K$, since all trajectories contract to the trajectory generated by the estimate system $\phi^{\mathrm{lse}}$. This observation completes the proof.
\end{proof}

\begin{remark}
	The constraint $\Gamma = \Gamma^T$ can be encoded in the optimization program by considering $\Gamma$ internally as represented by a triangular matrix on account of symmetricity, reducing the number of decision variables from $2+n^2$ to $2+n(n+1)/2$.
	Additionally, if $P = I$, then \eqref{eq:SOS-I} reduces to $x\transp (\mu I - P) x$. If $P$ is a diagonal matrix with nonzero elements on the diagonal, then $H := \frac{1}{2}(P \Gamma P\inv + \Gamma)$ is such that $\mathrm{diag}(H) = 0$ and $H_{i,j} = \frac{1}{2} (P_{i,i}/P_{jj} + 1)\Gamma_{i,j}$ for $i \neq j$.
\end{remark}

\begin{corollary}\label{cor:sos maximum contracting}
	Given the same hypotheses as in Theorem~\ref{thm:main synthesis sos}, consider the same SOS problem without the first inequality constraint involving $\mu$ in \eqref{eq:optimization problem}. The resulting gain $G$ is such that the closed-loop system is stable and contracting at rate $\gamma = \beta \mu + \varsigma$. In fact, the resulting closed-loop system is the least expansive ($\gamma > 0$) or most contractive ($\gamma < 0$) given the feedback structure considered. 
\end{corollary}

In the results given above, we have assumed access to a set $K$ in which we would like the resulting controller to be contracting and stabilizing. Often, this set might equal an \emph{outer-approximate reachable set} \cite{El-Kebir2023}. Key considerations of its choice are over what time span such a reachable set should be valid, and how this time span relates to the synthesis procedure, e.g., on account of lag due to computational time. 

We now proceed by considering an example scenario of an aerial vehicle experiencing an unknown wind field.

%
%
%

\section{Application}

In this section, we consider a numerical example based on a simplified model of the planar dynamics of an unmanned aerial vehicle (UAV) flying through a partially GPS-denied area while being affected by an unknown lateral wind field. In a GPS-denied area, we assume that the instantaneous position readings are unreliable, for instance due to a MEMS inertial measurement unit with poor noise properties. The nominal dynamics are as follows:
\begin{equation}\label{eq:uav}
	\dot{\vec{x}} = \frac{\dd}{\dd t}
	\begin{bmatrix}
		x \\
		y \\
		\dot{x} \\
		\dot{y}
	\end{bmatrix}
	=
	\begin{bmatrix}
	\dot{x} \\
	\dot{y} \\
	c_d (\dot{x}^2 + (\dot{y} + v_w(x))^2) \\
	c_w v_w (x)
	\end{bmatrix}
	+
	\begin{bmatrix}
		0 & 0 \\
		0 & 0 \\
		1 & 0 \\
		0 & 1
	\end{bmatrix}
	\begin{bmatrix}
		u_1 \\
		u_2
	\end{bmatrix},
\end{equation}
or $\dot{\vec{x}} = f(\vec{x}) + B \vec{u}$, where $v_w$ is an unknown polynomial modeling the lateral wind field as a position of the longitudinal position. 

For the weighted norm, we consider $P = I$. To account for the domain constraint $\vec{x} \in K$ in Theorem~\ref{thm:main synthesis sos} it would suffice to use a local optimization approach based on (S)DSOS \cite{Ahmadi2019}, since past approaches admit only global SOS solution. We instead opt to use a constrained sequential least squares programming (SLSQP) routine with nonlinear constraints, which converges reasonably fast; future work will rigorously consider runtime metrics.

Since the approach given in Theorem~\ref{thm:main synthesis sos} can lead to a gain matrix $\Gamma$ with values of exceedingly large magnitude, we include an additional constraint bounding the magnitude of each value of $\Gamma$, $\max_{ij} |\Gamma_{i,j}| \leq M$, where we take $M = 10$.

We take $p(\vec{x}) = x^2 + y^2 + \dot{x}^2 + \dot{y}^2$, and set $\alpha = 1000$ to satisfy the hypothesis of Lemma~\ref{lm:Rantzer Lyap}. To demonstrate our approach, we consider $K := \{[\begin{smallmatrix} 10 & 10 & 1 & 0.1 \end{smallmatrix}]\transp\} + \mathcal{B}_1^n$. Then, from Proposition~\ref{prop:sos oslip}, we find $\oslip_K (f_{\mathrm{nom}}) = 0.190$, where $f_{\mathrm{nom}}$ is defined by $f(\vec{x})$ from System~\eqref{eq:uav} when $v_w = 0$.

Solving program \eqref{eq:optimization problem} based on the relaxation given in Corollary~\ref{cor:sos maximum contracting}, we obtain $\mu_* = -0.141$. In addition, we obtain gain matrix 
$G = \left[\begin{smallmatrix}
		-3.142 &  0.015 & -3.365 &  0.473 \\
        -0.017 & -3.179 &  0.473 & -3.365
	\end{smallmatrix}\right]$.

We now show the level of disturbance that the current approach can handle based on a given matrix $N$ from \eqref{eq:N definition}, where $\Theta = \mathcal{Z}(N)$. We find the following condition for guaranteed $\gamma$-contraction and stabilization on $K$:
\begin{equation*}
	\gamma > -0.092 + \sqrt{-\lambda_{\max}(N | N_{22})/\lambda_{\max} (N_{22})}  \max_{x \in K} \Vert \DD b(x) \Vert_2.
\end{equation*}

\hil{In our simulations, we consider a uniform wind noise $w_v \in [-1, 1]$. We take $c_d = c_w = 10^{-2}$; the nominal trajectory is such that $v_w = 0$. The initial condition for the nominal system is $x_0 = [\begin{smallmatrix} 10 & 10 & 1 & 0.1 \end{smallmatrix}]\transp$, while the initial condition for the off-nominal system lies in $K$. We proceed by expressing the trajectory deviation in terms of the infinity norm, which is governed by the same contraction bound presented above.}

The resulting closed-loop trajectory deviation between the nominal system, and 20 realizations of possible off-nominal closed-loop systems, is shown in Fig.~\ref{fig:contraction}. 
To explicitly show outperformance of the contraction guarantee by our controller, we present each of the $(x, y)$ positions of the nominal and off-nominal systems in Fig.~\ref{fig:trajectories}. Contraction and stability can both be observed, where the contraction rate is further elucidated in Fig.~\ref{fig:contraction}.

\begin{figure}[h]
	\centering
	\includegraphics[width=0.8\linewidth]{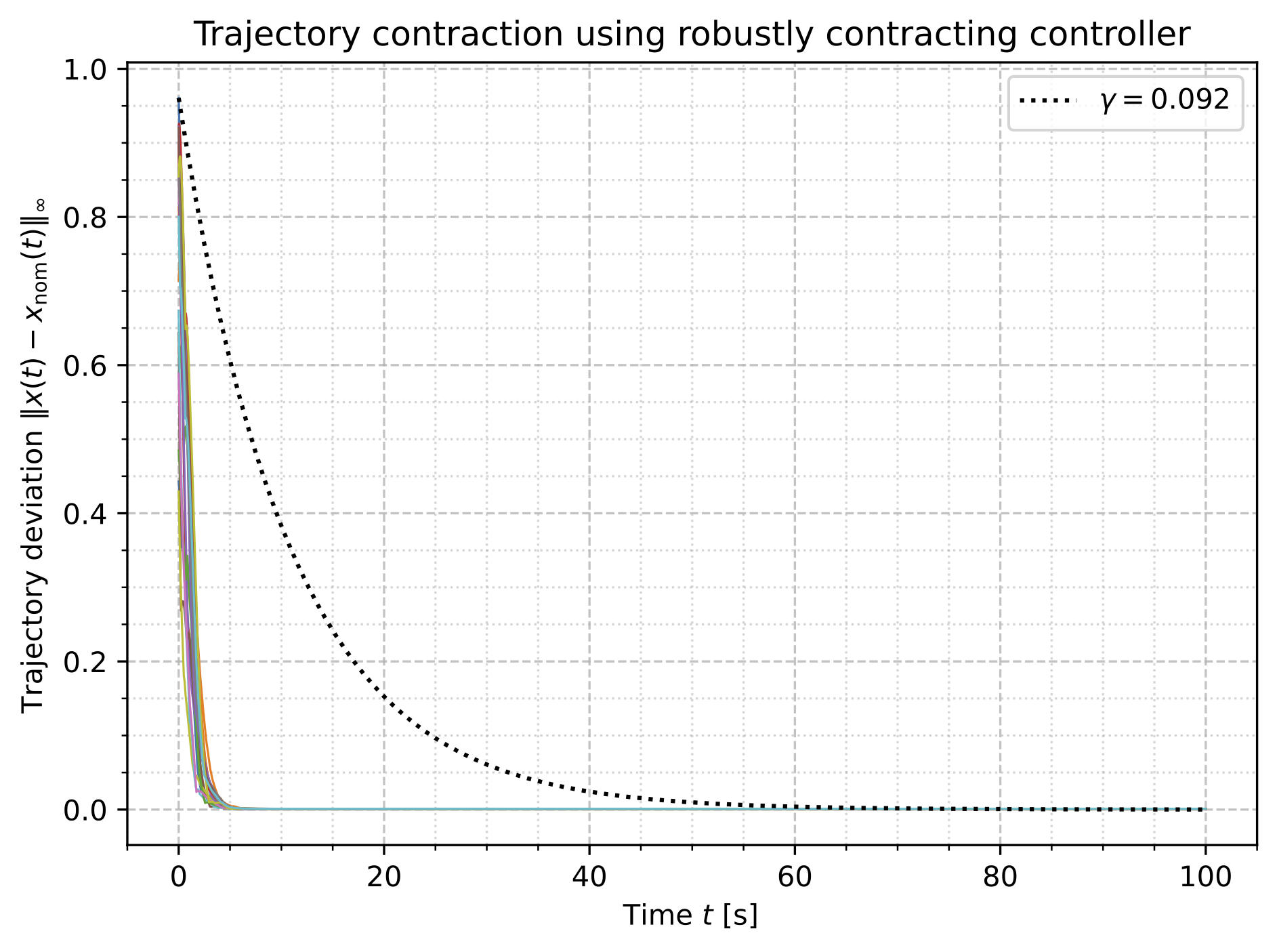}
	\caption{\hil{Trajectory deviation between the nominal trajectory and twenty realizations of the off-nominal system's trajectory with varying initial states. The theoretically guaranteed best contraction rate is honored and outperformed in practice using the method presented here.}}
	\label{fig:contraction}
\end{figure}

\begin{figure}[h]
	\centering
	\includegraphics[width=0.8\linewidth]{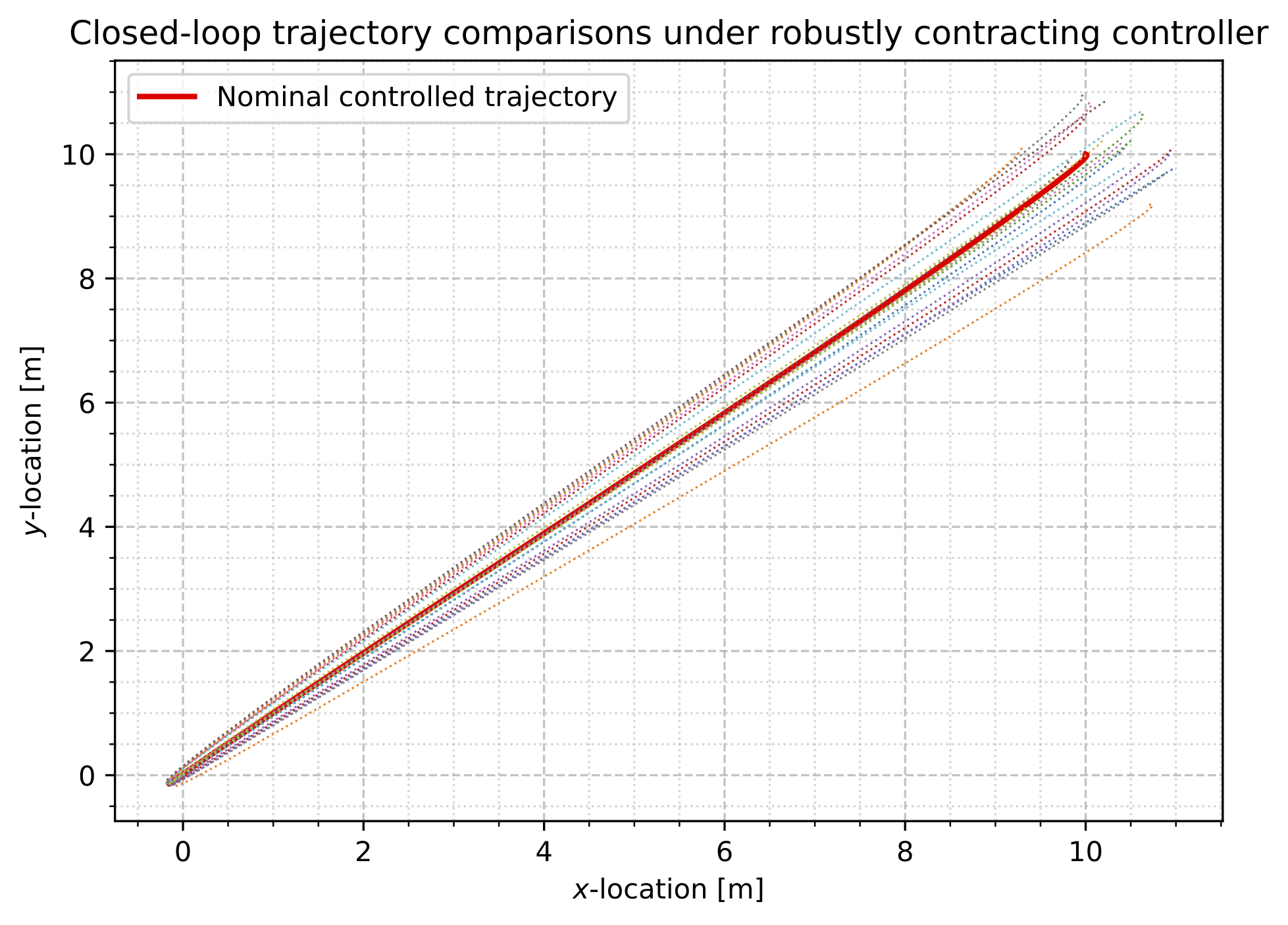}
	\caption{\hil{Trajectories of the nominal and off-nominal systems under the robustly contracting controller, where the first two states (position) are plotted. 
    }}
	\label{fig:trajectories}
\end{figure}

In future work, we shall present the performance of a complete closed-loop system incorporating a set-propagation-based observer structure.

\section{Conclusion}

In this work, we have presented a novel data-driven sum-of-squares-based approach to ensure local stability and contraction on a class of systems. We have shown how data-driven conditions can be developed to enable online controller synthesis based on recent advances in the cautious optimization and data informativity, as well as by leveraging a dual condition to Lyapunov's stability theorem. We have demonstrated the approach on an unmanned aerial vehicle system, showing how a fixed gain controller and associated contraction conditions can be developed directly from data.

\hil{Future work will include additional theory to show that the approach extends to general continuous nonlinear systems, provided that they can be approximated sufficiently well by a polynomial system. We also plan on extending the class of systems to those with nonlinear input maps, in addition to accounting for unknown output map switching as discussed in our prior work \cite{El-Kebir2024b}.}

\section*{Acknowledgments}

The authors thank Prof. Eduardo Sontag for his invaluable suggestions and fruitful discussions, and the anonymous reviewers for their insightful comments.




%
\bibliographystyle{IEEEtran}
\bibliography{root.bib}

\begin{thebibliography}{10}
\providecommand{\url}[1]{#1}
\csname url@samestyle\endcsname
\providecommand{\newblock}{\relax}
\providecommand{\bibinfo}[2]{#2}
\providecommand{\BIBentrySTDinterwordspacing}{\spaceskip=0pt\relax}
\providecommand{\BIBentryALTinterwordstretchfactor}{4}
\providecommand{\BIBentryALTinterwordspacing}{\spaceskip=\fontdimen2\font plus
\BIBentryALTinterwordstretchfactor\fontdimen3\font minus
  \fontdimen4\font\relax}
\providecommand{\BIBforeignlanguage}[2]{{%
\expandafter\ifx\csname l@#1\endcsname\relax
\typeout{** WARNING: IEEEtran.bst: No hyphenation pattern has been}%
\typeout{** loaded for the language `#1'. Using the pattern for}%
\typeout{** the default language instead.}%
\else
\language=\csname l@#1\endcsname
\fi
#2}}
\providecommand{\BIBdecl}{\relax}
\BIBdecl

\bibitem{Eising2023}
J.~Eising and J.~Cortes, ``Cautious optimization via data informativity,''
  2023, arXiv:2307.10232 [math.OC].

\bibitem{VanWaarde2020}
H.~J. Van~Waarde, J.~Eising, H.~L. Trentelman, and M.~K. Camlibel, ``Data
  informativity: {{A}} new perspective on data-driven analysis and control,''
  \emph{IEEE Transactions on Automatic Control}, vol.~65, no.~11, pp.
  4753--4768, 2020.

\bibitem{Lohmiller1998}
W.~Lohmiller and J.-J.~E. Slotine, ``On contraction analysis for non-linear
  systems,'' \emph{Automatica}, vol.~34, no.~6, pp. 683--696, 1998.

\bibitem{Lohmiller2000}
------, ``Nonlinear process control using contraction theory,'' \emph{AIChE
  Journal}, vol.~46, no.~3, pp. 588--596, 2000.

\bibitem{Aminzare2014a}
Z.~Aminzare and E.~D. Sontag, ``Contraction methods for nonlinear systems:
  {{A}} brief introduction and some open problems,'' in \emph{53rd {{IEEE
  Conference}} on {{Decision}} and {{Control}}}.\hskip 1em plus 0.5em minus
  0.4em\relax Los Angeles, CA, USA: IEEE, 2014, pp. 3835--3847.

\bibitem{Sontag2010a}
E.~D. Sontag, ``Contractive systems with inputs,'' in \emph{Perspectives in
  {{Mathematical System Theory}}, {{Control}}, and {{Signal Processing}}},
  M.~Morari, M.~Thoma, J.~C. Willems, S.~Hara, Y.~Ohta, and H.~Fujioka,
  Eds.\hskip 1em plus 0.5em minus 0.4em\relax Berlin: Springer Berlin
  Heidelberg, 2010, vol. 398, pp. 217--228.

\bibitem{Margaliot2016}
M.~Margaliot, E.~D. Sontag, and T.~Tuller, ``Contraction after small
  transients,'' \emph{Automatica}, vol.~67, pp. 178--184, 2016.

\bibitem{Tsukamoto2021}
H.~Tsukamoto and S.-J. Chung, ``Learning-based robust motion planning with
  guaranteed stability: {{A}} contraction theory approach,'' \emph{IEEE
  Robotics and Automation Letters}, vol.~6, no.~4, pp. 6164--6171, 2021.

\bibitem{Lopez2021}
B.~T. Lopez, J.-J.~E. Slotine, and J.~P. How, ``Robust adaptive control barrier
  functions: {{An}} adaptive and data-driven approach to safety,'' \emph{IEEE
  Control Systems Letters}, vol.~5, no.~3, pp. 1031--1036, 2021.

\bibitem{El-Kebir2024b}
H.~{El-Kebir}, M.~Ornik, Y.~K. Nakka, C.~Choi, and A.~Rahmani, ``Robust
  detection and identification of simultaneous sensor and actuator faults,'' in
  \emph{2024 {{IEEE Aerospace Conference}}}.\hskip 1em plus 0.5em minus
  0.4em\relax Big Sky, MT, USA: IEEE, 2024.

\bibitem{Balamurugan2016}
G.~Balamurugan, J.~Valarmathi, and V.~P.~S. Naidu, ``Survey on {{UAV}}
  navigation in {{GPS}} denied environments,'' in \emph{2016 {{International
  Conference}} on {{Signal Processing}}, {{Communication}}, {{Power}} and
  {{Embedded System}}}.\hskip 1em plus 0.5em minus 0.4em\relax Paralakhemundi,
  Odisha, India: IEEE, 2016, pp. 198--204.

\bibitem{Althoff2021a}
M.~Althoff and J.~J. Rath, ``Comparison of guaranteed state estimators for
  linear time-invariant systems,'' \emph{Automatica}, vol. 130, p. 109662,
  2021.

\bibitem{El-Kebir2023}
H.~{El-Kebir}, A.~Pirosmanishvili, and M.~Ornik, ``Online guaranteed reachable
  set approximation for systems with changed dynamics and control authority,''
  \emph{IEEE Transactions on Automatic Control}, vol.~69, no.~2, pp. 726--740,
  2024.

\bibitem{Davydov2022}
A.~Davydov, S.~Jafarpour, and F.~Bullo, ``Non-{{Euclidean}} contraction theory
  for robust nonlinear stability,'' \emph{IEEE Transactions on Automatic
  Control}, vol.~67, no.~12, pp. 6667--6681, 2022.

\bibitem{Rantzer2001}
A.~Rantzer, ``A dual to {{Lyapunov}}'s stability theorem,'' \emph{Systems \&
  Control Letters}, vol.~42, no.~3, pp. 161--168, 2001.

\bibitem{Zakeri2014}
H.~Zakeri and S.~Ozgoli, ``A sum of squares approach to robust {{PI}}
  controller synthesis for a class of polynomial multi-input multi-output
  nonlinear systems,'' \emph{Nonlinear Dynamics}, vol.~76, no.~2, pp.
  1485--1495, 2014.

\bibitem{Soderlind2006}
G.~S{\"o}derlind, ``The logarithmic norm. {{History}} and modern theory,''
  \emph{BIT Numerical Mathematics}, vol.~46, no.~3, pp. 631--652, 2006.

\bibitem{Motzkin1967}
T.~S. Motzkin, ``The arithmetic-geometric inequality,'' in \emph{1965
  {{Symposium}} of the {{Wright-Patterson Air Force Base}}}, 1967, pp.
  205--224.

\bibitem{Ahmadi2019}
A.~A. Ahmadi and A.~Majumdar, ``{{DSOS}} and {{SDSOS}} optimization: {{More}}
  tractable alternatives to sum of squares and semidefinite optimization,''
  \emph{SIAM Journal on Applied Algebra and Geometry}, vol.~3, no.~2, pp.
  193--230, 2019.

\bibitem{Demidovic1961}
B.~P. Demidovi{\v c}, ``Dissipativity of a nonlinear system of differential
  equations,'' \emph{Uspekhi Matematicheskikh Nauk}, vol.~16, no. 3(99), p.
  216, 1961.

\end{thebibliography}

%

%
%
%





\end{document}